\newif\iffinal
\documentclass[runningheads]{llncs}
\usepackage[T1]{fontenc}
\usepackage{graphicx}
\usepackage{amsmath}
\usepackage{tikz}
\usetikzlibrary{automata}
\usetikzlibrary{calc,decorations.pathmorphing,shapes}  % for squiggle arrows

\usepackage{ltl}
\usepackage{xspace}  % command \xspace automatically identifies when to put space

\usepackage{color}
\usepackage[colorlinks=true, linkcolor=black, citecolor=black]{hyperref}

\newcommand{\AP}{\mathsf{AP}}
\newcommand{\API}{{\mathsf{AP}_I}}
\newcommand{\APO}{{\mathsf{AP}_O}}

\newcommand{\LTLG}{\mathsf{G}}
\newcommand{\LTLF}{\mathsf{F}}
\newcommand{\LTLFG}{\mathsf{FG}}
\newcommand{\LTLGF}{\mathsf{GF}}
\newcommand{\LTLX}{\mathsf{X}}
\newcommand{\LTLR}{\mathsf{R}}
\newcommand{\LTLU}{\mathsf{U}}
\newcommand{\LTLW}{\mathsf{W}}
\newcommand{\NN}{\mathbb{N}}  % looks better
\newcommand\bbN{\NN}  % ak is used to using \bbN
\newcommand{\TRUE}{\mathsf{true}}
\newcommand{\FALSE}{\mathsf{false}}
\newcommand\rej{\ensuremath{\mathit{rej}}}
\newcommand\acc{\ensuremath{\mathit{acc}}}

\newcommand\mc\mathcal
\newcommand\EnfPre{\LTLsquare\LTLdiamond}
\newcommand\pto\rightharpoonup

\newcommand{\specialcellC}[2][c]{%
  \begin{tabular}[#1]{@{}c@{}}#2\end{tabular}%
}

\newcommand{\trans}[1]{\stackrel{{#1}}{\rightarrow}}

\newcommand\smallvdots{\scalebox{0.6}{\ensuremath{\vdots}}}

\newcommand{\ak}[1]{{\color{green!55!black!55}{\tt \scriptsize/#1/}}}

\newcommand\li{\begin{itemize}}
\newcommand\lo{\begin{enumerate}}
\renewcommand{\-}\item
\newcommand\il{\end{itemize}}
\newcommand\ol{\end{enumerate}}

\newcommand\reboot{$\tt reboot$\xspace}
\newcommand\slugs{$\tt slugs$\xspace}
\newcommand\strix{$\tt strix$\xspace}

\usepackage{stmaryrd}

\newcommand\x\times

\begin{document}
\title{Fully Generalized Reactivity(1) Synthesis\thanks{This work has been partially supported by the DFG through Grant No.~322591867 (GUISynth) and the BMWi through Grant No.~19A21026E (SafeWahr).\\ This is an extended version of the paper accepted to TACAS'24.}\\(full version)}
%\author{Author name(s) left out for double-blind submission}
\author{R\"udiger Ehlers and Ayrat Khalimov}
\authorrunning{R. Ehlers and A. Khalimov}
%\institute{Submitted to TACAS 2024}
\institute{Clausthal University of Technology, Germany\\
%\institute{Princeton University, Princeton NJ 08544, USA \and
%Springer Heidelberg, Tiergartenstr. 17, 69121 Heidelberg, Germany
\email{\{ruediger.ehlers,ayrat.khalimov\}@tu-clausthal.de}}
\maketitle              % typeset the header of the contribution
\vspace{-2mm}
\begin{abstract}
  Generalized Reactivity(1) (GR(1)) synthesis is a reactive synthesis approach in which the specification is split into two parts: a symbolic game graph, describing the safe transitions of a system, a liveness specification in a subset of Linear Temporal Logic (LTL) on top of it. Many specifications can naturally be written in this restricted form, and the restriction gives rise to a scalable synthesis procedure -- the reasons for the high popularity of the approach. For specifications even slightly beyond GR(1), however, the approach is inapplicable. This necessitates a transition to synthesizers for full LTL specifications, introducing a huge efficiency drop. This paper proposes a synthesis approach that smoothly bridges the efficiency gap from GR(1) to LTL by unifying synthesis for both classes of specifications. The approach leverages a recently introduced canonical representation of omega-regular languages based on a chain of good-for-games co-Büchi automata (COCOA). By constructing COCOA for the liveness part of a specification, we can then build a fixpoint formula that can be efficiently evaluated on the symbolic game graph. The COCOA-based synthesis approach outperforms standard approaches and retains the efficiency of GR(1) synthesis for specifications in GR(1) form and those with few non-GR(1) specification parts.
\end{abstract}
\section{Introduction}
%__Reactive synthesis__
Reactive synthesis is the process of automatically computing a provably correct reactive system from its formal specification~\cite{Church63}. A safety-critical system is often developed twice: first, when it is described using a formal specification, and second, when a system is implemented according to this specification. The dream of reactive synthesis is to fully eliminate manual implementation phase.

%__Reactive synthesis problem has a high complexity__
Reactive synthesis is however computationally hard. For specifications in the commonly used linear temporal logic (LTL), checking whether an implementation exists is 2EXPTIME-complete \cite{DBLP:conf/icalp/PnueliR89}.
The classical approach to solve reactive synthesis from LTL is to first translate the LTL formula into a deterministic parity automaton, followed by solving the induced two-player parity game \cite{HMC-RS}. The system player wins this game if and only if there is an implementation satisfying the specification. It is the first phase of translating LTL to parity automaton that usually represents a bottleneck. This observation spurred a series of synthesis approaches. For instance, in bounded synthesis, either the maximal number of states that a system can have \cite{DBLP:journals/sttt/FinkbeinerS13} or the longest system response time \cite{DBLP:conf/cav/FiliotJR09} is restricted. If there exists a system realizing the specification, then there exists one that adheres to some bounds, and bounded synthesis works well whenever small bounds suffice for realizing the given specification. Another approach is to synthesize implementations for parts of the specification, and to then compose them into one that realizes the whole specification \cite{KPV06,SS13,FJR11}. The approach of \cite{meyer2018strix} avoids constructing one large deterministic parity automaton and instead constructs many smaller ones that---when composed together---represent the original specification. Such decomposition proved beneficial on practical examples \cite{SyntComp23}. Finally, there are approaches that consider ``synthesis-friendly'' subsets of LTL. Alur and La Torre identified a number of such LTL fragments with a simpler synthesis problem \cite{DBLP:journals/tocl/AlurT04}, and this eventually led to the introduction of Generalized Reactivity(1) synthesis by Piterman et al.~\cite{PPS06}, GR(1) for short. GR(1)  synthesis gained a lot of prominence and was applied in domains such as robotics \cite{DBLP:conf/irc/WongK17,DBLP:conf/staf/GritznerG17}, cyber-physical system control \cite{10.1145/3513091,DBLP:conf/hybrid/WongpiromsarnTOXM11}, and chip component design \cite{DBLP:journals/jcss/BloemJPPS12,DBLP:journals/sttt/GodhalCH13}. We describe it in more detail.

%__GR(1) synthesis__
\looseness-1 In GR(1) synthesis, the specification is divided into two parts. The first part represents the \emph{safety properties} of a system and encodes a symbolic game graph. Each graph vertex encodes a valuation of last system inputs and outputs. The transitions in the graph represent how these variables can evolve in one step. For instance, a robot on a grid can move from its current cell to the left, right, up, or down, but cannot jump; this is easily encoded as a symbolic game graph. Secondly, there are \emph{liveness properties} of the following form: if certain vertices are visited infinitely often, then certain other vertices must be visited infinitely often as well. The liveness properties are encoded symbolically using LTL formulas of the shape $\bigwedge_i \LTLG\LTLF \varphi_i \to \bigwedge_j \LTLG\LTLF \psi_j$, where $\varphi_i$ and $\psi_j$ are Boolean formulas over input and output system propositions. Synthesis problems from many domains can be encoded naturally, or after some manual effort, into the GR(1) setting.

%__What makes GR(1) synthesis efficient__
\looseness-1 Constraining specifications to GR(1) form reduces the synthesis problem’s complexity from doubly-exponential to singly-exponential (in the number of propositions), or polynomial when the number of propositions is fixed \cite{DBLP:journals/jcss/BloemJPPS12}. The GR(1) synthesis problem can be solved by evaluating a fixpoint formula on the symbolic game arena. The fixpoint formula defines the set of vertices from which the system player satisfies the GR(1) liveness properties while staying in the game arena. The simple shape of GR(1) liveness properties makes the fixpoint formula simple. Moreover, evaluating the fixpoint formula on the symbolic game graph can be done efficiently using Binary Decision Diagrams (BDDs, \cite{DBLP:journals/tc/Bryant86}) as the underlying data structure. These factors together --- efficient implementation and relatively expressive specification language --- made GR(1) synthesis popular.

%__Drawback of GR(1) synthesis: no full LTL__
GR(1) synthesis has a drawback. A single property outside of GR(1) -- for instance, ``eventually the robot always stays in some stable zone'' ($\LTLF\LTLG\ \mathit{inStableZone}$) -- makes GR(1) synthesis inapplicable. Switching to full-LTL synthesizers introduces an abrupt efficiency drop, as they do not take advantage of the simple structure of GR(1)-like specifications. For improving the practical applicability of reactive synthesis, a synthesis approach exhibiting a smooth efficiency curve on the way from GR(1) to LTL would hence be useful.
While there are are some GR(1) synthesis extensions (e.g.,~\cite{DBLP:conf/fm/AmramMP19,DBLP:conf/nfm/Ehlers11}), they only extend it by certain specification classes and consequently do not support full LTL.

%__Our contribution__
This paper unifies synthesis for GR(1) and full LTL. Like in GR(1) synthesis, we aim at synthesis for specifications split into the safety part encoded as a symbolic game graph and the liveness part. Unlike the standard GR(1) synthesis, the liveness part can be any LTL or omega-regular property. For standard GR(1) specifications, our approach inherits the efficiency of GR(1) synthesis, including when a specification does not fall syntactically into this class, but is semantically a GR(1) specification. At the same time, for specifications that go beyond GR(1) and only have a few non-GR(1) components, our approach scales well.
% NOTE BY RUEDIGER: I've removed the following comment: \ak{R2: clarify what does `few' mean}.
% Reason for removal without addressing this: We wrote in the rebuttal that this cannot be made precise.

%__Technical contribution__
Our solution is based on the same fixpoint-evaluation-of-symbolic-game-graph idea. Our starting point is a folklore approach based on solving parity games by evaluating fixpoint equations~\cite{DBLP:journals/corr/BruseFL14}. We modify it so that it becomes applicable to specifications given 
in the form of a chain of good-for-games co-Büchi automata (COCOA).
Such chains have recently been proposed as a new canonical representation of omega-regular languages \cite{DBLP:conf/fsttcs/EhlersS22}, and it has been shown how minimal and canonical COCOA can be computed in polynomial time from a deterministic parity automaton of the language. Our COCOA-based synthesis approach converts the liveness part of the specification into a parity automaton, constructs the chain, builds the fixpoint formula from the chain, and finally evaluates it on the symbolic game graph. We show that the fixpoint formula built from the chain has a structure similar to GR(1) fixpoint formulas. This is not the case for the folklore approach via parity games, and as a result, our COCOA-based synthesizer is roughly an order of magnitude faster. The COCOA-based synthesis approach inherits the efficiency of GR(1) synthesis, and it is also efficient on specifications slightly beyond GR(1). Finally, our approach is the first application of the new canonical representation of omega-regular languages.

\section{Preliminaries}

\subsection*{Automata and languages}

Let $\NN = \{0,1,2,\ldots\}$ be the set of natural numbers including $0$.
Let $\AP$ be a set of \emph{atomic propositions};
$2^\AP$ denotes the \emph{valuations} of these propositions.
A Boolean formula represents a set of valuations:
for instance, $\bar a \land b$, also written $\bar a b$, encodes valuations
in which proposition $a$ has value $\FALSE$ and $b$ is $\TRUE$.
A Boolean function maps valuations of propositions to either $\TRUE$ or $\FALSE$.
Binary decision diagrams (BDDs) are a data structure for manipulating such functions.

A \emph{word} is a sequence of proposition valuations $w = x_0 x_1 \ldots \in (2^\AP)^\omega \cup (2^\AP)^*$.
A word can be finite or infinite.
%The suffix word $x_ix_{i+1}\ldots$ is also denoted by $w[i{:}]$.
A \emph{language} is a set of infinite words.
Given a language $L$,
the \emph{suffix language} of $L$ for some finite word $p \in (2^\AP)^*$
is $\mathcal{L}(L,p) = \{x_0 x_1 \ldots \in (2^\AP)^\omega \mid p \cdot x_0 x_1 \ldots \in L\}$.
The words in this set are called \emph{suffix words}.
% TODO: Check if OK...
The set of all suffix languages of $L$ is the set $\{\mc L(L,p) \mid p \in (2^\AP)^*\}$.
%For the simplicity of notation,
%we consider nested tuples $((a,b),c)$ to be equivalent to $(a,b,c)$ in this paper, whenever appropriate.
%Also, $\epsilon$ represents a special element that implicitly gets removed from tuples.

\looseness-1 Automata over infinite words are used to finitely represent languages.
We consider parity and co-Büchi automata with transition-based acceptance.
A \emph{parity automaton} is a tuple $\mc A = (\Sigma,Q,q_0,\delta)$
with a finite alphabet $\Sigma$ (usually $\Sigma = 2^\AP$),
a finite set of states $Q$,
an initial state $q_0 \in Q$,
and a finite transition relation $\delta \subseteq Q \times \Sigma \times Q \times \NN$
satisfying $(q,x,q',c) \in \delta \Rightarrow (q,x,q',c')\not\in \delta$ for all $q,x,q'$ and $c'\neq c$.
An automaton is \emph{complete} if
for every state $q$ and letter $x$ there exists at least one pair $(q',c)\in Q\x\NN$
s.t.\
$(q,x,q',c) \in \delta$;
it is \emph{deterministic} if exactly one such pair $(q',c)$ exists.
Wlog.\ we assume that automata are complete.
An automaton is \emph{co-Büchi} if only colors $1$ and $2$ occur in $\delta$,
and then we call the transitions with color $1$ \emph{rejecting} and those with color $2$ \emph{accepting}.

A \emph{run} of $\mc A$ on a word $w = x_0 x_1 \ldots \in \Sigma^\omega$
is a sequence $\pi = \pi_0 \pi_1 \ldots \in Q^\omega$
starting in $\pi_0 = q_0$ and such that $(\pi_i,x_i,\pi_{i+1},c_i) \in \delta$ for some $c_i$ for every $i \in \bbN$;
the induced \emph{color sequence} $c=c_0 c_1 \ldots$ is uniquely defined by $w$ and $\pi$.
A run is \emph{accepting}
if the lowest color occurring infinitely often in the induced color sequence is even (``min-even acceptance'').
When this minimal color is uniquely defined,
e.g.\ when there is only one accepting run,
it is called \emph{the color of $w$} wrt.\ $\mc A$.
A word is \emph{accepted} if it has an accepting run.
The automaton's language $\mathcal{L}(\mathcal{A})$ is the set of accepted words.
The language of the automaton $\mc A'$ derived from $\mc A$ by changing the initial state to $q$ is denoted by $\mc L(\mc A, q)$.

A \emph{co-Büchi language} is a language representable by a nondeterministic (equiv., deterministic) co-Büchi automaton.
The Co-B\"uchi languages are a strict subset of the omega-regular languages.

An automaton is \emph{good-for-games}
if there exists a strategy $f : \Sigma^* \to Q$ to resolve the nondeterminism to produce accepting runs on the accepted words, formally:
for every infinite word $w=x_0x_1\ldots$,
the sequence $\pi_0 \pi_1\ldots$
defined by $\pi_i = f(x_0 \ldots x_{i-1})$ for all $i\in\NN$
is a run, \emph{and} it is accepting whenever $w$ belongs to the language.

\subsection*{Games and our realizability problem}
\paragraph{LTL.}
A commonly used formalism to represent system specifications is \emph{Linear Temporal Logic} (LTL, \cite{DBLP:conf/focs/Pnueli77}).
It uses temporal operators $\LTLU$, $\LTLX$, and derived ones $\LTLG$ and $\LTLF$,
which we do not define here.
For details, we refer the reader to \cite{PP18}.

\paragraph{Games.}
An edge-labelled \emph{game} is a tuple
$G = (\AP_I,\AP_O,V,v_0,\delta,obj)$
where
$V$ is a finite set of vertices,
$v_0\in V$ is initial,
$\delta : V \x 2^\API\x 2^\APO \pto V$ is a partial function describing possible moves
(safety specification),
and $obj$ is a winning objective (liveness specification).
A \emph{play} is a maximal (finite or infinite) sequence of transitions
of the form
$(v_0,i_0,o_0,v_1)(v_1,i_1,o_1,v_2)(v_2,i_2,o_2,v_3)\ldots$;
the corresponding sequence $(i_0\cup o_0)(i_1 \cup o_1)\ldots$ is called the \emph{action sequence}.
An infinite play is winning for the system if it satisfies the objective $obj$;
when $obj$ is an LTL objective over $\API\cup\APO$,
the infinite play satisfies $obj$
iff the action sequence satisfies it.
A system \emph{strategy} is a function $f : (2^\API)^+ \to 2^\APO$.
The game is won by the system if it has a strategy $f$ such that
every play $(v_0,i_0, o_0,v_1)(v_1,i_1, o_1,v_2)\ldots$ is infinite and it satisfies the objective,
where $o_j = f(i_0 \ldots i_j )$ for all $j$.
To define parity games,
the winning objective $obj$ is set to be a parity-assigning function
$obj : V \to \bbN$,
and then an infinite play satisfies $obj$
iff
the minimal parity visited infinitely often in the sequence
$obj(v_0) obj(v_1) \ldots$
is even (min-even acceptance on states).

The \emph{enforceable predecessor operator} $\EnfPre$
reads a set of tuples $\Phi \subseteq 2^\AP \x V$ and
returns the set of positions
from which the system can enforce taking one of the transitions into the destination set:
\begin{equation}
  \EnfPre(\Phi) =
  \{
    v \in V \mid
    \forall i.\exists o:
    (i \cup o, \delta(v,i,o)) \in \Phi
  \}
\end{equation}
%\ak{define the `semantics' of winning in the game, i.e.,
%    if the game structure encodes safety language $\Psi$ and
%    the game LTL objective is $\Phi$,
%    then every inf action sequence satisfies $\Psi \to \Phi$, or smth like that}
%\re{Can't do that right now, as I will need to see where this is used. Doesn't sound like something that is strictly needed.}
% ak: agreed

\paragraph{Symbolic games with LTL objectives.}
Games can be represented symbolically.
For instance, the vertices can be encoded as valuations of Boolean variables $\AP$,
and transitions between the vertices can be encoded using a Boolean formula.
This paper focuses on solving symbolic games with LTL objectives:
\begin{quote}
  \emph{Given a symbolic game with LTL objective. Who wins the game?}
\end{quote}
The particular symbolic representation is not important as long as it provides
the operations for union, intersection, and complementation of sets of label-position tuples, and
the enforceable predecessor operator $\EnfPre$.
This paper focuses exclusively on the realizability problem;
the extraction of compact and efficient implementations merits a separate study.

\paragraph{Mu-calculus fixpoint formulas.}
For an introduction to using fixpoint formulas in synthesis,
we refer the reader to~\cite{HMC-RS},
and to~\cite{DBLP:reference/mc/BradfieldW18,niwinski2001rudiments} for mu-calculus in general.
The fixpoint formulas use the greatest ($\nu$) and least ($\mu$) fixpoint operators,
and the enforceable-predecessor operator $\EnfPre$.
% IMPLEMENTED: \ak{add example of win positions in safety/reachability games}
For instance,
the formula $\nu Y. \mu X. \EnfPre(Y \wedge (\overline{x} \vee X))$ represents the biggest set of vertices such that from all vertices in the set, the system can enforce that either $x$ does not hold along the next transition and this transition leads back to the same set, or the play gets closer to a position from which this can be enforced. This formula hence characterizes the positions from which the system can enforce that $\overline{x}$ holds infinitely often along a play.

\subsection*{Generalized Reactivity(1)}
Generalized Reactivity(1) is a class of assume-guarantee specifications
that includes safety and liveness components.
It gained popularity because many specifications naturally fall into the GR(1) class,
and the restricted nature of GR(1) admits an efficient synthesis approach.
For the purpose of this paper,
we define a GR(1) specification as a game
$G_\text{gr1} = (\API,\APO,V,v_0,\delta,\Phi)$
with an
LTL winning objective of the form
$\Phi = \bigwedge_{i=1}^m \LTLGF a_i \to \bigwedge_{j=1}^n \LTLGF g_j$,
where each assumption $a_i$ and guarantee $g_j$
are Boolean formulas over $\AP_I\cup\AP_O$.
The original GR(1) specification class~\cite{PPS06} uses logical formulas to describe the symbolic arena.

\subsection*{Solving GR(1) games using fixpoints}
We now show how to solve GR(1) games by evaluating fixpoint formulas on GR(1) game arenas.
Consider a GR(1) game $G_\text{gr1}=(\API,\APO,V,v_0,\delta,\Phi)$ with
$\Phi=\bigwedge_{i=1}^m \LTLGF a_i \to \bigwedge_{j=1}^n \LTLGF g_j$.
The set of positions $W \subseteq V$ from which the system player wins the game is characterized by the fixpoint equation~\cite{DBLP:conf/cav/EhlersR16,DBLP:journals/jcss/BloemJPPS12}:
\begin{equation} \label{eqn:fixpoint-GR1}
  W = \nu Z. \bigwedge_{j=1}^n \mu Y. \bigvee_{i=1}^m \nu X. \LTLsquare\LTLdiamond \big[(g_j \wedge Z) ~\lor~ Y ~\lor~ (\neg a_i \wedge X)\big]
\end{equation}
This fixpoint formula ensures that
the system chooses to move into states of one of the three kinds:
(1) states where it waits for an environment goal $a_i$ to be reached, possibly forever
    ($\neg a_i \land X$),
(2) states that move the system closer to reaching its goal number $j$ ($Y$), or
(3) winning states that satisfy system goal number $j$ $(g_j \land Z)$.
The conjunction over all guarantees to the right of $\nu Z$ ensures
that all liveness guarantees are satisfied from all winning positions
(unless some environment liveness assumption is violated).
The disjunction over the environment goals permits the system
to wait for the satisfaction of any of the environment liveness goals.
At the end of evaluating the fixpoint formula,
$Z$ consists of the winning positions for the system.
The system wins the GR(1) game
if and only if $W$ includes $v_0$.

\paragraph{Example.}
Consider a GR(1) game
with
$\API=\{u\}$, $\APO=\{x,y\}$,
and $\Phi = \LTLGF u \to (\LTLGF x \land \LTLGF y)$.
Equation~\ref{eqn:fixpoint-GR1} becomes:
\begin{equation} \label{eqn:exampleFixpoingGR1}
W =
\nu Z .
\left[
  \begin{matrix}
    \mu Y. \nu X. \LTLsquare \LTLdiamond (xZ ~\vee~ Y ~\vee~ \bar u X ) ~\land\,\\
    \mu Y. \nu X. \LTLsquare \LTLdiamond (yZ ~\vee~ Y ~\vee~ \bar u X ) ~~~~
\end{matrix}
  \right]
\end{equation}
For conciseness,
we write $xZ$ instead of $x \land Z$,
and $\bar a$ instead of $\neg a$.

\subsection*{Solving symbolic parity games using fixpoints}
Consider a parity game $(\API,\APO,V,v_0,\delta,c)$ with colors $\{0,\ldots,n\}$.
The winning positions for the system player in such game are characterized by the fixpoint formula from~\cite{W02,DBLP:journals/corr/BruseFL14} adapted to our setting:
\begin{equation} \label{eqn:parityFixpoint}
  W = \nu X^0 \mu X^1 \ldots \sigma X^n. \LTLsquare \LTLdiamond (\vee_{i=1}^{n}\mathsf{color}_i \wedge X^i)
\end{equation}
The operators $\nu$ and $\mu$ alternate,
so the symbol $\sigma$ is $\mu$ if $n$ is odd and $\nu$ if $n$ is even;
${\sf color}_i = \{v \mid c(v) = i\}$ denotes the set of vertices of color $i$.

\subsection*{Solving symbolic LTL games using fixpoints}
\label{page:solving-ltl-games-via-parity}

Let $G$ be a game with LTL objective $\Phi$.
We can construct a deterministic parity automaton $\mc A$ for $\Phi$,
build the product parity game $G \otimes \mc A$, and
solve it with the help of Equation~\ref{eqn:parityFixpoint}.
An alternative approach is to embed the product into the fixpoint formula
by using vector notation \cite{DBLP:reference/mc/BradfieldW18}.
%Resolved: \ak{R1: separate the example into separate paragraph}

Consider an example.
Let $G=(\API,\APO,V,v_0,\delta,\Phi)$ be a game with $\Phi = \LTLGF u \to (\LTLGF x \land \LTLGF y)$.
The parity automaton for $\Phi$ is shown on Figure~\ref{fig:gr1Example-parity}.
It has two states, $q_0$ and $q_1$, and uses three colors.
For three colors,
the parity fixpoint formula in Equation~\ref{eqn:parityFixpoint} has structure $\nu Z.\mu Y.\nu X.$
We index each set variable with the state of the automaton,
thus $Z$ is split into $Z_0$ and $Z_1$, etc.
The formula is:
\begin{equation} \label{eqn:exampleParityFixpoint}
\left[ \begin{matrix} W_0 \\ W_1 \end{matrix} \right] =
\nu \left[ \begin{matrix} Z_0 \\ Z_1 \end{matrix} \right]\!\!.
\mu \left[ \begin{matrix} Y_0 \\ Y_1 \end{matrix} \right]\!\!.
\nu \left[ \begin{matrix} X_0 \\ X_1 \end{matrix} \right]\!\!.
\LTLsquare \LTLdiamond\!\left[
\begin{matrix}
   x Z_1 ~\vee~ \bar x u Y_0 ~\vee~ \bar x \bar u X_0  \\
   y Z_0 ~\vee~ \bar y u Y_1 ~\vee~ \bar y \bar u X_1
\end{matrix}
\right]
\end{equation}
The top row encodes the transitions from state $q_0$ of the parity automaton:
$q_0 \trans{x:0} q_1$ becomes $xZ_1$,
$q_0 \trans{\bar xu:1} q_1$ becomes $\bar x u Y_1$,
$q_0 \trans{\bar x\bar u:2} q_0$ becomes $\bar x \bar u X_0$.
After formula evaluation,
the variable $W_0$ contains game positions winning for the system
wrt.\ the parity automaton $\mc A_{q_0}$, while $W_1$ does so wrt.\ $\mc A_{q_1}$.

\begin{figure}[tb]
  \centering
  \begin{tikzpicture}
    \node[thick,state,inner sep=0pt,minimum size=0.7cm,fill=black!10!white] (qA) at (-6,-2) {$q_0$};
    \node[thick,state,inner sep=0pt,minimum size=0.7cm,fill=black!10!white] (qB) at (-4,-2) {$q_1$};

    \draw[->,thick] (qA) to[loop left] node[left] {$\bar x \bar u{:}2$} (qA);
    \draw[->,thick] (qA) to[loop below] node[below] {$\bar x u{:} 1$} (qA);
    \draw[->,thick] (qA) to[bend left=10] node[above] {$x{:}0$} (qB);

    \draw[->,thick] (qB) to[loop right] node[right] {$\bar y \bar u{:}2$} (qB);
    \draw[->,thick] (qB) to[loop below] node[below] {$\bar y u{:}1$} (qB);
    \draw[->,thick] (qB) to[bend left=10] node[below] {$y{:}0$} (qA);

    \draw[->,thick] (-6.7,-1.6) -- (qA);
    \draw[fill=black] (-6.7,-1.6) circle (0.05cm);
  \end{tikzpicture}
  \vspace{-2mm}
  \caption{Parity automaton for $\LTLGF u \to (\LTLGF x \land \LTLGF y)$. Transitions are labeled by the proposition valuations for which they can be taken as well as the color of the transition.}
  % RESOLVED: \ak{R1: notation unclear; describe $\delta$}
  \label{fig:gr1Example-parity}
\end{figure}
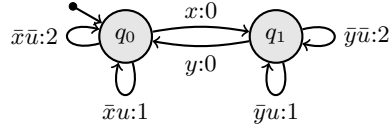

In general,
suppose we are given a game whose winning objective is a deterministic parity automaton $(2^\AP,Q,q_0,\delta)$ with transition function $\delta : Q\x\Sigma \to Q\x \bbN$ that uses $n$ colors $\{0,\ldots,n-1\}$.
The set of winning game positions is characterized by the fixpoint formula:
\begin{equation}\label{eqn:fixpoint-parity}
  \left[\begin{matrix} W_1 \\ \smallvdots \\ W_{|Q|}\end{matrix}\right] =
  \nu \left[ \begin{matrix} X_1^0 \\ \smallvdots \\ X_{|Q|}^0 \end{matrix} \right]\!\!.
  \mu \left[ \begin{matrix} X_1^1 \\ \smallvdots \\ X_{|Q|}^1 \end{matrix} \right] \ldots~
  \sigma \left[ \begin{matrix} X_1^{n-1} \\ \smallvdots \\ X_{|Q|}^{n-1} \end{matrix} \right]\!\!.
  \LTLsquare \LTLdiamond\!
  \left[
  \begin{matrix}
    \psi_1 \\
    \smallvdots \\
    \psi_{|Q|}
  \end{matrix}
  \right]
\vspace{-1mm}
\end{equation}
\begin{equation*}
\text{where for all } j \in \{1, \ldots, |Q|\}, \text{we have }  \psi_j = \bigvee_{\scalebox{0.8}{\specialcellC{$x \in 2^\AP$\vspace{-0.5mm}\\$\text{let } (q,c)=\delta(q_j,x)$}}}
  x \land X_{ind(q)}^c
\end{equation*}
where $ind: Q \to \{1,\ldots,|Q|\}$ is some state numbering (one-to-one) that maps the initial automaton state $q_0$ to $1$.
The game is won by the system if and only if the initial game position belongs to $W_1$.

% RESOLVED:
% The following does not hold - reachability works as deterministic parity automata need to be complete, so that we have a "trap" state. 
% The G \alpha \rightarrow \Phi is, anyway, not exactly GR(1) because if some safety assumptions or guarantee is not ever violated, it matters which player violates it *first*.
% \ak{note that the LTL objective must be infinitary (or works for reachability as well?):
% with reachability does not work we do not have the semantics $\LTLG \alpha \to \Phi$.}

\section{Chains of Good-for-Games co-Büchi Automata}
\label{sec:COCOA}
This section reviews the chain of good-for-games co-Büchi automata representation \cite{DBLP:conf/fsttcs/EhlersS22} for $\omega$-regular languages used by our synthesis approach in Section~\ref{sec:synthesisFromCOCOA}.

Like parity automata,
a chain of co-Büchi automaton representation of a language assigns colors to words.
The central difference is that the chain representation relies on a sequence of automata,
each taking care of a single color. % in a parity automaton representation of the language.

\begin{definition}
  Let $L \subseteq \Sigma^\omega$ be an omega-regular language.
  A falling chain of languages $L_1 \supset L_2 \supset \ldots \supset L_n$ is a \emph{chain-of-co-Büchi representation} of $L$ if
\begin{itemize}
\item every language $L_i$ for $i \in \{1,\ldots,n\}$ is a co-Büchi language, and
\item for every $w \in \Sigma^\omega$,
      the word $w$ is in $L$ if and only if $w \not\in L_1$ or the highest index $i$ such that $w \in L_i $ is even.
\end{itemize}
\end{definition}
\emph{Examples.}
The universal language $\Sigma^\omega$ has the singleton-chain $L_1 = \emptyset$,
and the empty language has the chain $(L_1=\Sigma^\omega) \supset (L_2 = \emptyset)$.
The language of the LTL formula $\LTLG\LTLF a$ over a single atomic proposition $a$
is expressed by the chain $(L_1 = L(\LTLF\LTLG \overline{a})) \supset (L_2 = \emptyset)$,
and $L(\LTLF\LTLG a)$ by $(L_1 = \Sigma^\omega) \supset (L_2 = L(\LTLF\LTLG a)) \supset (L_3 = \emptyset)$.

%Every deterministic parity automaton with colors $0$, $1$, $\ldots$, $n$ can be represented by a chain $A_1 \supset A_2 \supset \ldots A_n$ where $A_1$ accepts the words with color at least $1$, $A_2$ accepts the words with color at least $2$, and so on. These languages are all co-Büchi languages as they reason about which transitions in a deterministic parity automaton are taken only finitely often. % AK: more details are needed

The definition of the natural color of a word from \cite{DBLP:conf/fsttcs/EhlersS22} provides a canonical way to represent $L$ as a chain of co-Büchi languages  $L_1 \supset L_2 \supset \ldots \supset L_n$,
which uses the minimal number of colors.
Moreover,
Abu Radi and Kupferman describe a procedure to construct a minimal and canonical good-for-games co-B\"uchi automaton for a given co-B\"uchi language \cite{DBLP:journals/lmcs/RadiK22}.
Thus, every omega-regular language has a canonical minimal chain-of-co-B\"uchi-automata representation (\emph{\textbf{COCOA}}).

The canonization procedure in \cite[Thm.4.7]{DBLP:journals/lmcs/RadiK22} ensures the following property.
\begin{lemma}[\cite{DBLP:journals/lmcs/RadiK22}] \label{lem:canon-gfg-prop}
  Fix a canonical GFG co-B\"uchi automaton $\cal A$ computed by \cite[Thm.4.7]{DBLP:journals/lmcs/RadiK22}.
  For every state $q$ and letter $x$, either there is
  \li
  \- exactly one accepting transition, or there are
  \- one or more rejecting transitions.
     In this case:
     \li
     \- all successors of $q$ on $x$ share the same suffix language $L'$,
        i.e., for every two successors $s_1$ and $s_2$ of $q$ on $x$: $L(\mc A, s_1) = L(\mc A, s_2)$, and
     \- for every state $q'$ with suffix language $L'$, there is a rejecting transition to $q'$ from $q$ on $x$.
     \il
  \il
\end{lemma}
\begin{corollary}[\cite{DBLP:journals/lmcs/RadiK22}] \label{cor:lem:canon-gfg-prop}
  Fix a canonical GFG co-B\"uchi automaton $\mc A$.
  For every finite prefix $p$,
  all states in which $\mc A$ ends after reading $p$ share the same language.
\end{corollary}

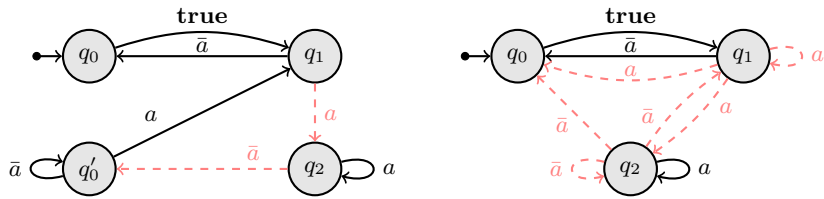
\begin{figure}[tb]
\centering
\vspace{-1mm}
\begin{tikzpicture}
\node[thick,state,inner sep=0pt,minimum size=0.7cm,fill=black!10!white] (qA) at (-6,-2) {$q_0$};
\node[thick,state,inner sep=0pt,minimum size=0.7cm,fill=black!10!white] (qB) at (-3,-2) {$q_1$};
\node[thick,state,inner sep=0pt,minimum size=0.7cm,fill=black!10!white] (qC) at (-3,-3.5) {$q_2$};
\node[thick,state,inner sep=0pt,minimum size=0.7cm,fill=black!10!white] (qD) at (-6,-3.5) {$q'_0$};

\draw[->,thick] (qA) to[out=18,in=162] node[above] {$\mathbf{true}$} (qB);
\draw[->,thick] (qB) to node[above=-0.6mm] {$\bar a$} (qA);

\draw[->,thick] (qC) to[loop right] node[right] {$a$} (qC);
\draw[->,thick] (qD) to[loop left] node[left] {$\bar a$} (qD);
\draw[->,thick] (qD) to node[above left,pos=0.3] {$a$} (qB);

% Rejecting transitions
%\draw[->,thick,dashed,color=red!50!white] (qB) to[loop above] node[above] {$c$} (qB);
\draw[->,thick,dashed,color=red!50!white] (qB) to[bend right=0] node[above right,pos=0.8] {$a$} (qC);
%\draw[->,thick,dashed,color=red!50!white] (qC) to[bend right=10] node[below right=-0.5mm,pos=0.7] {$a$} (qB);
\draw[->,thick,dashed,color=red!50!white] (qC) to node[above,pos=0.2] {$\bar{a}$} (qD);

\draw[->,thick] (-6.7,-2) -- (qA);
\draw[fill=black] (-6.7,-2) circle (0.05cm);
%\draw (current bounding box.north west) rectangle (current bounding box.south east);
\end{tikzpicture}
$\quad \quad$%
\begin{tikzpicture}
\node[thick,state,inner sep=0pt,minimum size=0.7cm,fill=black!10!white] (qA) at (-6,-2) {$q_0$};
\node[thick,state,inner sep=0pt,minimum size=0.7cm,fill=black!10!white] (qB) at (-3,-2) {$q_1$};
\node[thick,state,inner sep=0pt,minimum size=0.7cm,fill=black!10!white] (qC) at (-4.5,-3.5) {$q_2$};

\draw[->,thick] (qA) to[out=18,in=162] node[above] {$\mathbf{true}$} (qB);
\draw[->,thick] (qB) to[bend left=0] node[above=-0.6mm] {$\bar a$} (qA);

\draw[->,thick] (qC) to[loop right] node[right] {$a$} (qC);

% Rejecting transitions
\draw[->,thick,dashed,color=red!50!white] (qB) to[bend left=18] node[above=-0.6mm] {$a$} (qA);
\draw[->,thick,dashed,color=red!50!white] (qB) to[loop right] node[right] {$a$} (qB);
\draw[->,thick,dashed,color=red!50!white] (qB) to[bend left=10] node[below right=-0.5mm,pos=0.2] {$a$} (qC);
\draw[->,thick,dashed,color=red!50!white] (qC) to[bend left=10] node[above left=-0.5mm,pos=0.2] {$\bar a$} (qB);
\draw[->,thick,dashed,color=red!50!white] (qC) to node[below left=-0.5mm,pos=0.5] {$\bar a$} (qA);
\draw[->,thick,dashed,color=red!50!white] (qC) to[loop left,in=195,out=165,looseness=8] node[left] {$\bar a$} (qC);

\draw[->,thick] (-6.7,-2) -- (qA);
\draw[fill=black] (-6.7,-2) circle (0.05cm);
%\draw (current bounding box.north west) rectangle (current bounding box.south east);
\end{tikzpicture}
\vspace{-2mm}
  \caption{%
    Example of a deterministic co-Büchi automaton (on the left) and an equivalent canonical good-for-games co-Büchi automaton for the same language.
    Both automata have minimal numbers of states.
    %The one of the right merges the role of the state $q'_0$ with $q_0$ and $q_1$, at the cost of introducing non-determinism that is however good-for-games.
    Rejecting transitions are dashed.
    GFG co-B\"uchi automata are exponentially more succinct than deterministic ones~\cite{DBLP:conf/icalp/KuperbergS15}.}
\label{fig:examplecoBuchi}
\end{figure}
Figure~\ref{fig:examplecoBuchi} shows an example of a canonical good-for-games co-Büchi automaton,
and Figure~\ref{fig:gr1Example} on page~\pageref{fig:gr1Example} an example of a COCOA.

\subsection*{Strategies to get back on the track}

Every GFG automaton has a strategy to resolve its nondeterminism
such that a word is accepted if and only if the run adhering to this strategy is accepting.
We allow such strategies to diverge for a finite number of steps,
and show that this divergence does not affect the acceptance by canonical GFG automata.

Given a COCOA $\mc A^1,\ldots,\mc A^n$,
define the \emph{natural color} of a word to be the largest level $l$ such that
$\mc A^l$ accepts the word,
or $0$ if no such $l$ exists.
Thus, a word is accepted by the COCOA if and only if the natural color is even.

\paragraph{GFGness strategies $f^l$.}
Let $f^l : \Sigma^* \to Q^l$ be a GFG witness resolving nondeterminism in $\mc A^l$,
for every $l \in \{1,\ldots,n\}$;
we call $f^l$ a golden strategy of $\mc A^l$,
and the induced run for some given word is called its \emph{golden run}.

\paragraph{Restrictions $g^l$.}
The synthesis approach, which will be described later, considers combined runs of all automata.
Its efficiency depends on the number of reachable states in $Q^1 \x \ldots \x Q^n$,
so it is beneficial to reduce this number.
To this end,
we introduce a restriction on successor choices.
We first define a helpful notion:
for a co-B\"uchi automaton $\mc A$ and its state $q$,
let $L^{acc}(q)$ denote the set of words which have a run from $q$ visiting only accepting transitions.
For several automata $\mc A^1,\ldots,\mc A^l$ and their states $q^1,\ldots,q^l$,
define $L^{acc}(q^1,\ldots,q^l) = \bigwedge_i L^{acc}(q^i)$.
Then, for $l \in \{1,\ldots,n\}$, define a \emph{restriction} function
$g^l : Q^l \x \Sigma \x Q^1 \x \ldots \x Q^{l-1} \to 2^{Q^l}$:
for every $q^l$, $x$, $r^1,\ldots,r^{l-1}$,
let
$g^l(q^l, x, r^1,\ldots,r^{l-1}) = S \subseteq \delta^l(q^l,x)$
be a maximal set such that
for every $r^l \in S$
there exists no other $\tilde r^l \in S$ with $L^{acc}(r^1,\ldots,r^{l-1},\tilde r^l) \subseteq L^{acc}(r^1,\ldots,r^l)$.
\label{def:restriction-function}
Intuitively,
given a current state $q^l$ of the automaton $\mc A^l$,
a letter $x$, and successor states $r^1,\ldots,r^{l-1}$ of the automata on lower levels,
the function $g^l$ returns a set of states among which $\mc A^l$ should pick a successor.
Runs $\rho^1 = q^1_0 q^1_1\ldots,\ldots,\rho^n=q^n_0 q^n_1 \ldots$ of $\mc A^1,\ldots,\mc A^n$
on a word $x_0 x_1 \ldots$
\emph{satisfy}
restrictions $g^1,\ldots,g^n$
if for every level $l \in \{1,\ldots,n\}$ and step $i \in \bbN$:
$q^l_{i+1} \in g^l(q^l_i,x_i,q^1_{i+1},\ldots,q^{l-1}_{i+1})$.
Strategies $f^l : \Sigma^* \to Q^l$ for $l \in \{1,\ldots,n\}$ \emph{satisfy} restrictions $g^1,\ldots,g^n$ if on every word the strategies yield runs satisfying the restrictions.

The following lemma states that requiring runs of $\mc A^1,\ldots,\mc A^n$
to satisfy the restrictions $g^1,\ldots,g^n$ preserves the natural colors and the GFGness.

\begin{lemma}
  \label{lem:better-gfg-strategies}
  There exist strategies $f^l : \Sigma^* \to Q^l$ for $l \in \{1,\ldots,n\}$ satisfying the restrictions $g^1,\ldots,g^n$ such that
  for every word of a natural color $c$,
  the strategies yield accepting runs $\rho^1,\ldots,\rho^c$ of $\mc A^1,\ldots,\mc A^c$.
\end{lemma}

\begin{proof}
  Fix a word $w$ of a natural color $c$.
  Each automaton $\mc A^l$ of the chain has a GFG witness in the form of a strategy $h^l : \Sigma^* \to Q^l$ to resolve nondeterminism.
  From such strategies and the restrictions $g^1,\ldots,g^n$,
  we construct the sought strategies $f^1,\ldots,f^n$,
  inductively on the level, starting from the smallest level $1$ and proceeding upwards to $n$.

  Fix $l \in \{1,\ldots,n\}$,
  and suppose the strategies $f^1,\ldots,f^{l-1}$ are already defined;
  we define the strategy $f^l: \Sigma^* \to Q^l$.
  Fix a moment $i-1$.
  Let
  $q^l_{i-1}$ be the state of the run $\rho^l$ proceeding according to $f^l$,
  $\tilde q^l_i=h^l(x_0\ldots x_{i-1})$ the successor state in the original run $\tilde \rho^l$ according to $h^l$,
  $q^1_i,\ldots,q^{l-1}_i$ the successor states in $\rho^1,\ldots,\rho^{l-1}$ adhering to $f^1,\ldots,f^{l-1}$, and
  $Q^l_i = g^l(q^l_{i-1},x_{i-1},q^1_i,\ldots,q^{l-1}_i)$ the allowed successors on level $l$.
  Then:
  \li
  \- if $Q^l_i=\{q^l_i\}$ describes a unique choice, then $f^l(x_0\ldots x_{i-1})=q^l_i$ takes it,
  \- else $f^l$ picks any $q^l_i \in Q^l_i$ s.t.\ $L^\acc(q^1_i,\ldots,q^{l-1}_i,q^l_i) \supseteq L^\acc(q^1_i,\ldots,q^{l-1}_i,\tilde q^l_i)$.
     Note that such $q^l_i$ always exists
     because
     in canonical GFG co-B\"uchi automata a choice of a nondeterministic transition does not narrow the subsequent nondeterminism resolution.
  \il

  We now show that the strategies $f^1,\ldots,f^l$ preserve the natural colors.
  Fix a word $w$.
  It suffices to prove that the original strategy $h^l$ yields an accepting run $\tilde \rho^l$
  if and only if
  $f^l$ yields an accepting run $\rho^l$.
  If $\tilde \rho^l$ is rejecting,
  then $\rho^l$ is also rejecting,
  for $h^l$ is a witness of GFGness.
  Now assume that $\tilde \rho^l$ is accepting.
  After some moment $m$,
  the runs $\rho^1,\ldots,\rho^{l-1},\tilde\rho^l$ never make a rejecting transition,
  hence $w_m w_{m+1} \ldots \in L^\acc(q^1_m,\ldots,q^{l-1}_m,\tilde q^l_m)$.
  Let $m' \geq m$ be the first moment after $m$ when $\rho^l$ visits a rejecting transition;
  if no such $m'$ exists, we are done.
  At moment $m'$, the strategy $f^l$ picks a successor $q^l_{m'+1}$ such that $L^\acc(q^1_{m'+1},\ldots,q^l_{m'+1}) \supseteq L^\acc(q^1_{m'+1},\ldots,\tilde q^l_{m'+1})$.
  Since $w_{m'{+}1}\ldots \in L^\acc(q^1_{m'+1},\ldots,q^{l-1}_{m'+1},\tilde q^l_{m'+1})$,
  that suffix also belongs to a larger $L^\acc$ wrt.\ $q^l_{m'+1}$.
  Hence the run $\rho^l$ is accepting.
  \qed
\end{proof}

\paragraph{Get-back strategies $f^l_\star$.}
\label{def:get-back-on-the-track}
We now consider runs that diverge from golden runs.
Given an individual strategy $f^l : \Sigma^* \to Q^l$,
define $f^l_\star : \Sigma^* \x Q^l \x \Sigma \pto Q^l$ to be a strategy-like function
which, when presented with a choice, makes the same choice as $f^l$.
Formally:
for every $p \in \Sigma^*$, $q \in Q^l$ reachable from the initial state on reading $p$, and $x \in \Sigma$,
the value $f^l_\star(p,q,x) = f^l(p\cdot x)$ if $\mc A^l$ needs to take a rejecting transition from $q$ on $x$,
otherwise there is no choice to be made and $f^l_\star(p,q,x) = q'$ for the unique successor $q'$ of $q$ on reading $x$.
It follows from properties of canonical GFG automata (Lemma~\ref{lem:canon-gfg-prop})
that every successor chosen by $f^l_\star$ satisfies the transition relation of $\mc A^l$.
We now prove that it is sufficent to adhere to $f^l_\star$ only eventually.

\begin{lemma}\label{lem:cocoa-runs}
  Fix a COCOA and a word $w$.
  For $l\in \{1,\ldots,n\}$,
  suppose $\mc A^l$ on $w$ has a rejecting run $\rho^l$
  that eventually adheres to $f^l_\star$,
  where $f^l_\star$ is constructed from $f^l$ of Lemma~\ref{lem:better-gfg-strategies}.
  Then $\mc A^l$ rejects $w$.
\end{lemma}
The proof is based on Lemma~\ref{lem:canon-gfg-prop},
which implies that two diverging runs of a canonical GFG automaton on the same word
can always be converged once a rejecting transition is taken.

\begin{proof}
  For $l=0$ the claim trivially holds; assume $l>0$.
  Let $\rho_\star^l$ be the golden run of $\mc A^l$ on the word.
  Let $m$ be the moment starting from which $\rho^l$ adheres to the golden strategy of $\mc A^l$.
  Let $n$ be the first moment $n \geq m$
  when $\mc A^l$ makes a rejecting transition:
  by properties of canonical GFG automata (Lemma~\ref{lem:canon-gfg-prop}),
  there must be a rejecting transition to the same state as in $\rho_\star^l$.
  The strategy $f_\star^l$ moves the automaton $\mc A^l$ in $\rho^l$ into the same state at moment $n+1$ as it is in $\rho_\star^l$.
  Afterwards,
  the strategy $f_\star^l$ ensures that $\mc A^l$ in $\rho^l$ follows exactly the same transitions as $\mc A^l$ in $\rho_\star^l$.
  Hence,
  the golden run $\rho_\star^l$ is rejecting:
  $\mc A^l$ rejects $w$. \qed
\end{proof}

\normalsize
\subsection*{COCOA product}
In this section,
we compose individual automata of COCOA into a product
which is a good-for-games alternating parity automaton~\cite{AlternatingGFG}.
The results %of the previous section (NOTE: it's the same section)
above
imply that the languages of a COCOA and its product coincide.
Later we use COCOA products to solve games with LTL objectives.

\paragraph{Alternating automata.}
A \emph{simple%
\footnote{`Simple' refers to a simpler form of the transition function. We use $\delta : Q \x \Sigma \to 2^Q \x \bbN \x \{\rej,\acc\}$ while the general form is $\delta : Q \x \Sigma \to {\sf B}^+(Q)$ plus parity assignment $Q\x\Sigma\x Q \to \bbN$. We forbid mixing conjunctions and disjunctions.}%
alternating parity automaton}
$(\Sigma,Q,q_0,\delta)$
has a transition function of type $\delta : Q \x \Sigma \to 2^Q \x \bbN \x \{\rej,\acc\}$.
For instance,
$\delta(q,x) = (\{q_1,q_2\},1,\rej)$
means that
from state $q$ on reading letter $x$
there are transitions to $q_1$ and $q_2$,
both labelled with color $1$,
and the choice between $q_1$ and $q_2$ is controlled by the rejector player.
There are two players, rejector and acceptor,
and the acceptance of a word $w = x_0 x_1 \ldots$ is defined via the following \emph{word-checking game}.
Starting in $q_0$,
the two players resolve nondeterminism and build a \emph{play}
$(q_0,c_0,pl_0,q_1) (q_1,c_1,pl_1,q_2)\ldots$:
suppose the play sequence is in state $q_i$,
let $\delta(q_i,x_i) = (Q_{i+1},c_i,pl_i)$:
if $pl_i = \rej$ then the rejector chooses a state $q_{i+1} \in Q_{i+1}$, otherwise the acceptor chooses.
The play sequence is then extended by $(q_i,c_i,pl_i,q_{i+1})$ and the procedure repeats from state $q_{i+1}$.
The play is \emph{won} by the acceptor if the minimal color appearing infinitely often in $c_0 c_1\ldots$ is even
(min-even acceptance),
otherwise it is won by the rejector.
The word-checking game is \emph{won} by the acceptor
if
it has a strategy $f_w : Q^* \to Q$ to resolve its nondeterminism to win every play;
otherwise the game is won by the rejector, who then also has a winning strategy.
Note that although the acceptor strategy does not know the rejector choices beforehand,
it knows the whole word $w$.
The word is \emph{accepted} by the automaton if the word-checking game is won by the acceptor.

A simple alternating automaton is \emph{good-for-games}, abbreviated \emph{A-GFG},
if the acceptor player has a strategy
$f_\text{acc} : (Q \x \Sigma)^* \to Q$
to win the word-checking game for every accepting word, and
the rejector player has a strategy
$f_\text{rej} : (Q\x\Sigma)^* \to Q$
winning for every rejected word.
These strategies depend only on the currently seen word prefix, not the whole word.
%From the strategies $f$ and $f'$,
%for every word $w$,
%we can derive $f_w$ and $f'_w$ as:
%$f_w(q_0 \ldots q_m) = f((q_0,w_0)\ldots (q_m,w_m))$, similarly for $f'_w$.
We remark that our definition of GFGness differs from~\cite{AlternatingGFG}
but they show the equivalence~\cite[Thm.8]{AlternatingGFG}.

\paragraph{COCOA product.}
The product is built in three steps.
First, we define a naive product,
which combines individual chain automata into A-GFG in a straightforward way.
The naive product may contain states whose removal does not affect its language,
hence in the second step we define a product with reduced sets of states and transitions.
In turn, the reduced product may miss transitions beneficial for synthesis.
Therefore,
in the last step, we enrich the reduced product with transitions to derive the optimized, and final, COCOA product.

Given a COCOA $\mc A^l = (\Sigma,Q^l,q_0^l,\delta^l)$ with $l\in \{1,\ldots,n\}$, the
\emph{naive COCOA product} is
the following simple alternating parity automaton $(\Sigma,Q,q_0,\delta)$.
Each state is a tuple from $Q^1\x \ldots \x Q^n$,
$q_0 = (q^1_0,\ldots,q^n_0)$,
and the set of states consists of those reachable from the initial state under the transition relation defined next.
The transition relation $\delta : Q \x \Sigma \to 2^Q\x\bbN\x\{\rej,\acc\}$
simulates individual automata of the COCOA.
Consider an arbitrary $(q^1,\ldots,q^n)\in Q$, $x \in \Sigma$;
let $r$ be the smallest number such that
$\mc A^r$ has a rejecting transition from $q^r$ on reading $x$,
i.e., $(q^r,x,\tilde q^r,1) \in \delta^r$ for some $\tilde q^r \in Q^r$,
otherwise set $r$ to $n+1$.
By abuse of notation,
define $\delta^l(q^l,x) = \{\tilde q^l \mid \exists p: (q^l,x,\tilde q^l,p) \in \delta^l\}$
to be the set of successor states of $q^l$ on reading $x$ in $\mc A^l$.
Let $pl^r$ be $\rej$ for odd $r$ and $\acc$ for even $r$.
Then,
$\delta((q^1,\ldots,q^n),x) = (\tilde Q ,r-1, pl^r)$,
where:
\begin{equation*}
  \tilde Q =
  \{
    (\tilde q^1,\ldots,\tilde q^n) \mid
    \tilde q^l \in \delta^l(q^l,x) \textit{ for every $l$}
  \}.
\end{equation*}
Notice that the automata on levels $l < r$
have unique successors ($\tilde q^l$ is unique) as their transitions are accepting and hence deterministic (by Lemma~\ref{lem:canon-gfg-prop} on page~\pageref{lem:canon-gfg-prop}).
The automata on levels $l\geq r$ may need to resolve nondeterminism,
which is done by a single player $pl^r$ in the product.

The \emph{reduced COCOA product} $(\Sigma, Q_R, q_0,\delta_R)$ is defined by replacing the definition of $\tilde Q$ by
\begin{equation*}
  \tilde Q = \{ (\tilde q^1,\ldots,\tilde q^n) \mid \tilde q^l \in g^l(\tilde q^1,\ldots,\tilde q^{l-1},x,q^l) \textit{ for every $l$}\}
\end{equation*}
where the restriction function $g^l$ was defined on page~\pageref{def:restriction-function}.
As a result,
this set $\tilde Q$ has no two states $(q^1,\ldots,q^n)$ and $(\tilde q^1,\ldots,\tilde q^n)$
with $L^\acc(q^1,\ldots,q^n) \subseteq L^\acc(\tilde q^1,\ldots,\allowbreak{}\tilde q^n)$.
The set $Q_R$ of states of the reduced COCOA product is the set of states from $Q^1\x\ldots\x Q^n$ reachable under the above definition.

Finally, given a reduced COCOA product $(\Sigma,Q_R,q_0,\delta_R)$,
we now define the \emph{optimized COCOA product} $(\Sigma,Q_O,q_0,\delta_O)$.
It has the same states $Q_O=Q_R$ as the reduced product but more transitions.
For $(q^1,\ldots,q^n) \in Q_O$, $x \in \Sigma$,
let $(\tilde Q_R,r-1,pl^r) = \delta_R((q^1,\ldots,q^n),x)$.
Then $\delta_O((q^1,\ldots,q^n),x) = (\tilde Q_O,r-1,pl^r)$,
where
\begin{align*}
  \tilde Q_O =
    \big\{
      &(\tilde q^1,\ldots,\tilde q^n) \in Q_O : \\
      &\forall l \in \{1,\ldots,r-1\}{:}~ q^l \in \delta^l(q^l,x) ~\land \\
      &\forall l \in \{r,\ldots,n\}.\exists \tilde{\tilde q}^l \in \delta^l(q^l,x){:}~ L(\tilde q^l) = L(\tilde{\tilde q}^l)
    \big\}.
\end{align*}
In the first condition, the successor $q^l$ for $l\leq r{-}1$ is uniquely defined.
The second condition on levels higher than $r-1$ allows for state jumping,
even when an automaton on such a level makes a deterministic transition.
%Intuitively,
%the optimized product has the states of the reduced product but transitions of the naive product and a bit more.

\begin{lemma}\label{lem:cocoa-product}
  For every COCOA,
  the optimized product is A-GFG and has the same language as the COCOA.
\end{lemma}

\begin{proof}
  We describe two strategies,
  $f_\text{acc} : (Q\x\Sigma)^* \to Q$ for the acceptor and $f_\text{rej} : (Q\x\Sigma)^* \to Q$ for the rejector,
  and prove two claims: for every word,
  \lo
  \- if the word is accepted by COCOA, the acceptor wins the word-checking game using $f_\text{acc}$,
  \- if the word is rejected by COCOA, the rejector wins the word-checking game using $f_\text{rej}$.
  \ol
  The lemma follows from these claims.

  We define $f_\text{acc}$.
  Given a finite history $h = ((q^1_1,...,q^n_1), x_1) ... ((q^1_i,...,q^n_i),x_i)$,
  let $f_\text{acc}(h) = (q^1_{i+1},...,q^n_{i+1})$,
  where for $l=1,...,n$:
  \li
  \- if $l$ is even: $q^l_{i+1} = f_\star^l(x_1 \ldots x_{i-1},q^l_i,x_i)$;
  \- if $l$ is odd, pick arbitrary $q^l_{i+1} \in g^l(q^1_{i+1},\ldots,q^{l-1}_{i+1},q^l_i)$.
  \il
  The strategy $f_\text{rej}$ is built similarly but $f^l_\star$ is used for odd $l$.

  We now prove the first item using contraposition.
  Suppose for a given word,
  the acceptor equipped with $f_\text{acc}$ loses the word-checking game.
  Then there exists a play consisent with $f_\text{acc}$ such that
  for some odd $l$
  it encodes accepting runs $\rho^1,\ldots,\rho^l$ of $\mc A^1,\ldots,\mc A^l$ and a rejecting run $\rho^{l+1}$ of $\mc A^{l+1}$.
  After some moment $m$,
  the first $l$ runs never visit a rejecting transition,
  and hence $\rho^{l+1}$ starting with the moment $m$ always follows the strategy $f_\star^{l+1}$.
  Lemma~\ref{lem:cocoa-runs} implies that that the word has the odd natural color $l$,
  and therefore is rejected by the COCOA.
  % Not quite true:
  % Lemma~\ref{lem:cocoa-runs} talks about a rejecting *run*, but here we do not necessarily get a *run*
  % because we allowed to jump on higher level even when the higher-level automaton makes an _accepting_ transition
  % This is minor though: the Lemma can be easily fixed.
  This concludes the proof of the first item, the second item is proven similarly.
  \qed
\end{proof}

\paragraph{Example.}
\begin{figure}[tb]
\centering
\begin{tikzpicture}
\node[thick,state,inner sep=0pt,minimum size=0.7cm,fill=black!10!white] (qA) at (-6,-2) {$q_0$};
\node[thick,state,inner sep=0pt,minimum size=0.7cm,fill=black!10!white] (qB) at (-4,-2) {$q_1$};

\draw[->,thick] (qA) to[loop left] node[left] {$\bar x$} (qA);
\draw[->,thick] (qB) to[loop right] node[right] {$\bar y$} (qB);

\draw[->,thick,dashed,color=red!50!white] (qA) to[loop below, looseness=7] node[below] {$x$} (qA);
\draw[->,thick,dashed,color=red!50!white] (qB) to[loop below, looseness=7] node[below] {$y$} (qB);

\node at (-5,-2.9) {$\mathcal{A}^1$};
  \node at (-5,-3.4) {$\LTLFG \bar x \lor \LTLFG \bar y$};

% Rejecting transitions
%\draw[->,thick,dashed,color=red!50!white] (qB) to[loop above] node[above] {$c$} (qB);
\draw[->,thick,dashed,color=red!50!white] (qA) to[bend left=10] node[above] {$x$} (qB);
\draw[->,thick,dashed,color=red!50!white] (qB) to[bend left=10] node[below] {$y$} (qA);

\draw[->,thick] (-6.7,-1.6) -- (qA);
\draw[fill=black] (-6.7,-1.6) circle (0.05cm);
\end{tikzpicture}
$\quad \quad$
\begin{tikzpicture}
\node[thick,state,inner sep=0pt,minimum size=0.7cm,fill=black!10!white] (qA) at (-6,-2) {$p_0$};
\node[thick,state,inner sep=0pt,minimum size=0.7cm,fill=black!10!white] (qB) at (-4,-2) {$p_1$};

\draw[->,thick] (qA) to[loop left] node[left] {$\bar x  \bar u$} (qA);
\draw[->,thick] (qB) to[loop right] node[right] {$\bar y \bar u$} (qB);

\draw[->,thick,dashed,color=red!50!white] (qA) to[loop below, looseness=7] node[below] {$x \vee u$} (qA);
\draw[->,thick,dashed,color=red!50!white] (qB) to[loop below, looseness=7] node[below] {$y \vee u$} (qB);

\node at (-5,-2.9) {$\mathcal{A}^2$};
  \node at (-5,-3.4) {$\LTLFG \bar x \bar u \lor \LTLFG \bar y \bar u$};

% Rejecting transitions
%\draw[->,thick,dashed,color=red!50!white] (qB) to[loop above] node[above] {$c$} (qB);
\draw[->,thick,dashed,color=red!50!white] (qA) to[bend left=10] node[above] {$x \vee u$} (qB);
\draw[->,thick,dashed,color=red!50!white] (qB) to[bend left=10] node[below] {$y \vee u$} (qA);

\draw[->,thick] (-6.7,-1.6) -- (qA);
\draw[fill=black] (-6.7,-1.6) circle (0.05cm);
\end{tikzpicture}
\vspace{-2mm}
\caption{COCOA for the language $\LTLGF u \to (\LTLGF x \land \LTLGF y)$. Rejecting transitions are dashed.}
\label{fig:gr1Example}
\end{figure}
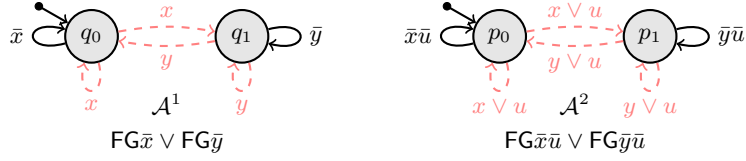
Figure~\ref{fig:gr1Example-cocoaProduct} shows the optimized product for COCOA in Figure~\ref{fig:gr1Example} for $\LTLGF u \to (\LTLGF x \land \LTLGF y)$.
We now intuitively explain which states can be removed without affecting the product language.
The COCOA encodes a falling chain of languages,
hence the strongly connected components of $\mc A^2$ cannot accept any word that is rejected by $\mc A^1$.
Thus, every strongly connected component of $\mc A^2$ consisting only of accepting transitions,
which monitors the liveness assumption,
also needs to monitor one of the liveness guarantees.
This means that when building the product between $\mc A^1$ and $\mc A^2$,
we obtain unnecessary product elements.
Consider an arbitrary suffix word on which
the automata $\mc A^1$ and $\mc A^2$ when starting from $q_0$ and $p_1$ visit only accepting transitions:
when starting from $q_0$ and $p_0$, we also visit only accepting transitions on that suffix word.
Since the acceptor player tries to visit only accepting transitions,
there is never an incentive for the acceptor to choose $p_1$ for resolving the nondeterminism
when $\mc A^1$ is in state $q_0$.
This makes the choice of $p_0$ better than that of $p_1$ when $\mc A^1$ stays in $q_0$.
Similarly, transiting from $(q_1,p_1)$ to $(q_1,p_1)$ is better for the acceptor than transiting to $(q_1,p_0)$.
Thus, the combinations $(q_0,p_1)$ and $(q_1,p_0)$ are unnecessary.
\begin{figure}[tb]
  \centering
  \begin{tikzpicture}
    \node[thick,state,inner sep=1.5mm,minimum size=5mm,fill=black!10!white,rectangle,rounded corners] (00) at (-6,-2) {$q_0,p_0$};
    \node[thick,state,inner sep=1.5mm,minimum size=5mm,fill=black!10!white,rectangle,rounded corners] (11) at (-3,-2) {$q_1,p_1$};
    \node[circle, inner sep=0pt, minimum size=0pt] (00to11) at (-4.5,-1.855) {};
    \node[circle, inner sep=0pt, minimum size=0pt] (11to00) at (-4.5,-2.145) {};

    \draw[->,thick] (00) to[loop left,looseness=5] node[left] {\specialcellC{$\bar x \bar u{:}2$\vspace{-1.3mm}\\{\scalebox{0.8}{\tt rej}}}} (00);
    \draw[->,thick] (00) to[loop below] node[below] {\specialcellC{$\bar x u{:}1$\vspace{-1.3mm}\\{\scalebox{0.8}{\tt acc}}}} (00);
    \draw[-,thick] (00.10) to[bend right=10]  (00to11);
    \draw[->,thick] (00to11) to[bend left=30] node[above,xshift=-2.3mm,yshift=-1.1mm] {\specialcellC{$x{:}0$\vspace{-1.3mm}\\{\scalebox{0.8}{\tt rej}}}} (11.165);
    \draw[->,thick] (00to11) to[bend right=90] (00.45);

    \draw[->,thick] (11) to[loop right,looseness=5] node[right] {\specialcellC{$\bar y \bar u{:}2$\vspace{-1.3mm}\\\scalebox{0.8}{\tt rej}}} (11);
    \draw[->,thick] (11) to[loop below] node[below] {\specialcellC{$\bar y u{:}1$\vspace{-1.3mm}\\\scalebox{0.8}{\tt acc}}} (11);
    \draw[-,thick] (11.190) to[bend right=10]  (11to00);
    \draw[->,thick] (11to00) to[bend left=30] node[below,xshift=2.3mm,yshift=1.1mm] {\specialcellC{$y{:}0$\vspace{-1.3mm}\\{\scalebox{0.8}{\tt rej}}}} (00.345);
    \draw[->,thick] (11to00) to[bend right=90] (11.225);

    \draw[->,thick] (-6.7,-1.6) -- (00);
    \draw[fill=black] (-6.7,-1.6) circle (0.05cm);
  \end{tikzpicture}
  \vspace{-2mm}
  \caption{Optimized COCOA product for $\LTLGF u \to (\LTLGF x \land \LTLGF y)$.
           It has only two nondeterministic transitions, connecting $(q_0,p_0)$ and $(q_1,p_1)$,
           controlled by the rejector.
	   For instance, $\delta((q_0, p_0),x) = (\{(q_0,p_0),(q_1,p_1)\},0,\rej)$.}
  \label{fig:gr1Example-cocoaProduct}
\end{figure}
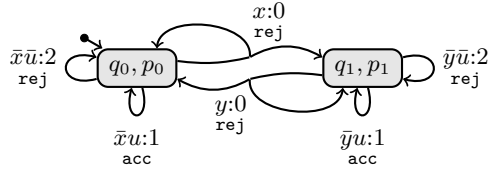

\section{Solving LTL Games Using Chain of co-Büchi Automata}
\label{sec:synthesisFromCOCOA}
This section shows how to solve symbolic games with LTL objectives by going through COCOA.
For a given LTL specification we construct a deterministic parity automaton
and then a COCOA using the effective procedure of~\cite{DBLP:conf/fsttcs/EhlersS22}.
We then compute the COCOA product.
Finally, we encode the symbolic game with a COCOA product objective into a fixpoint formula.
The latter step is simple because the COCOA product is a good-for-games alternating automaton,
and such automata are composable with games~\cite[Thm.8]{AlternatingGFG}.
Finally, we show that the GR(1) fixpoint equation is a special case of the COCOA fixpoint formula.

\subsection*{Fixpoint formula for games with COCOA objectives}
Given a game with an objective in the form of an optimized COCOA product
$(2^\AP,Q,q_0,\delta)$,
we construct a fixpoint formula that characterizes the set of winning positions.
Since the COCOA product is a good-for-games parity automaton,
the formula resembles Equation~\ref{eqn:fixpoint-parity}.
It has the structure $\nu X^0.\mu X^1.\ldots\sigma X^n$
where $n+1$ is the number of colors in the COCOA product,
and the operators $\nu$ and $\mu$ alternate.
As before, we use the vector notation,
and split each variable $X^l$ into $|Q|$ variables $X_1^l,\ldots,X_{|Q|}^l$,
one per state of the COCOA product,
and the $k$th row in the fixpoint formula encodes transitions from state $q_k$ of the product.
Let $ind: Q \to \{1,\ldots,|Q|\}$ be some one-to-one state numbering
with the initial state of the COCOA product mapped to $1$,
and let $\raisebox{-0.3mm}{\scalebox{1.15}{\sf OP}}^{pl}$ denote $\bigvee$ when $pl$ is $\acc$
otherwise it is $\bigwedge$.
The following fixpoint formula computes,
for each state $q$ of the COCOA product,
the set $W_{ind(q)}$ of game positions
from which the system player wins the game wrt.\ the COCOA product whose initial state is set to $q$:
\begin{align}\label{eqn:fixpoint-COCOA}
  \left[ \begin{matrix} W_1 \\ \smallvdots \\ W_{|Q|} \end{matrix}\right]
  &=
  \nu \left[ \begin{matrix} X_1^0 \\ \smallvdots \\ X_{|Q|}^0 \end{matrix} \right]\!\!.
  \mu \left[ \begin{matrix} X_1^1 \\ \smallvdots \\ X_{|Q|}^1 \end{matrix} \right] \ldots~
  \sigma \left[ \begin{matrix} X_1^n \\ \smallvdots \\ X_{|Q|}^n \end{matrix} \right]\!\!.
  \LTLsquare \LTLdiamond\!
  \left[
  \begin{matrix}
    \psi_1 \\
    \smallvdots \\
    \psi_{|Q|}
  \end{matrix}
  \right]\!,\, \text{ where for all $j$:}\\\nonumber
  \psi_j &= \hspace{-6mm}\bigvee_{\scalebox{0.8}{\specialcellC{$x \in 2^\AP$\vspace{-0.5mm}\\$\text{let } (\tilde Q,c,pl)=\delta(q_j,x)$}}} \hspace{-6mm}
  \left(
  x \land
  \raisebox{-0.35mm}{\scalebox{1.2}{\sf OP}}^{pl}_{\!\!q \in \tilde Q} X_{ind(q)}^c
  \right)
\end{align}
The game wrt.\ the COCOA product is won by the system player if and only if $v_0 \in W_1$.
Since the languages of COCOA and its optimized product coincide (Lemma~\ref{lem:cocoa-product}),
we arrive at the following theorem.

\begin{theorem}
  A game with an LTL objective $\Phi$ is won by the system if and only if the initial game position belongs
  to $W_1$ computed by Equation~\ref{eqn:fixpoint-COCOA} for the optimized COCOA product for $\Phi$.
\end{theorem}

\paragraph{Example.}
Consider the LTL specification $\LTLGF u \to (\LTLGF x \land \LTLGF y)$.
The optimized product contains only states $(q_0,p_0)$ and $(q_1,p_1)$.
The fixpoint formula is:
%%%Original:
%%%\begin{align}
%%%\label{eqn:simpleGR1ExampleMultiDimensionalFixpoint}
%%%  \nu\!\! \left[ \begin{matrix} Z_{00} \\ Z_{01} \\ Z_{10} \\ Z_{11} \end{matrix} \right]\!\!.
%%%  \mu \!\!\left[ \begin{matrix} Y_{00} \\ Y_{01} \\ Y_{10} \\ Y_{11} \end{matrix} \right]\!\!.
%%%  \nu\!\! \left[ \begin{matrix} X_{00} \\ X_{01} \\ X_{10} \\ X_{11} \end{matrix} \right]\!\!.
%%%\LTLsquare \LTLdiamond\!\!\left[
%%%\begin{matrix}
%%%   x Z_\textit{all} ~\vee~ \bar x u (Y_{00} \lor Y_{01}) ~\vee~ \bar{x}\bar{u} X_{00}  \\
%%%  ~x Z_\textit{all} ~\vee~ \bar x (y \vee u) (Y_{00} \lor Y_{01}) ~\vee~ \bar x \bar y \bar u X_{01} \\
%%%  ~y Z_\textit{all} ~\vee~ \bar y(x \vee u) (Y_{10} \lor Y_{11}) ~\vee~ \bar{x}\bar{y}\bar{u} X_{10}  \\
%%%  y Z_\textit{all} ~\vee~ \bar y u (Y_{10} \lor Y_{11}) ~\vee~ \bar{y}\bar{u} X_{11}
%%%\end{matrix}
%%%\right]
%%%\end{align}
%%%%
%%%For the case of optimized COCOA product, we obtain the following simplified fixpoint formula:
\begin{equation*}
\nu \left[ \begin{matrix} Z_{00} \\ Z_{11} \end{matrix} \right]\!\!.
\mu \left[ \begin{matrix} Y_{00} \\ Y_{11} \end{matrix} \right]\!\!.
\nu \left[ \begin{matrix} X_{00} \\ X_{11} \end{matrix} \right]\!\!.
\LTLsquare \LTLdiamond\!\left[
\begin{matrix}
   x Z_{00} Z_{11} ~\vee~ \bar x u Y_{00} ~\vee~ \bar x \bar u X_{00}  \\
   y Z_{00} Z_{11} ~\vee~ \bar y u Y_{11} ~\vee~ \bar y \bar u X_{11}
\end{matrix}
\right]
\end{equation*}
where the subscript index $ij$ denotes a state $(q_i,p_j)$ of the optimized COCOA product.
The LTL game is won by the system if and only if at the end of evaluation the initial game position $v_0$ belongs to $Z_{00}$.
This formula has a structure similar to the GR(1) Equation~\ref{eqn:exampleFixpoingGR1},
in particular it uses the conjunction over $Z$ variables
which leads to a reduction of the number of fixpoint iterations.
In contrast,
the parity formula in Equation~\ref{eqn:exampleParityFixpoint} misses this acceleration.

\subsection*{GR(1) synthesis as a special case}
We argue that for GR(1) specifications,
the COCOA fixpoint Equation~\ref{eqn:fixpoint-COCOA} becomes similar -- in spirit -- to GR(1) fixpoint Equation~\ref{eqn:fixpoint-GR1}.
Consider a GR(1) formula $\bigwedge_{i=1}^m \LTLGF a_i \to \bigwedge_{j=1}^n \LTLGF g_j$.
Its COCOA has two automata, $\mc A^1$ and $\mc A^2$.
The automaton $\mc A^1$ accepts exactly the words that violate one of the guarantees,
while $\mc A^2$ accepts exactly the words that violate one of the guarantees and one of the assumptions.
In order to reason able number of states in canonical automata, we assume henceforth that in the GR(1) formula, no assumption implies another assumption or guarantee, and no guarantee implies another guarantee.
The structures of $\mc A^1$ and $A^2$ are as follows.
\li
\- $\mc A^1$ has one state per guarantee ($n$ in total),
   while $\mc A^2$ has one per combination of liveness assumption and guarantee ($m \cdot n$ in total).
\- The optimized COCOA product has exactly one state for each assumption-guarantee combination,
   $m\cdot n$ in total, versus $n\cdot (m\cdot n)$ for the non-optimized product.
   Let $\{1,\ldots,m\}\x\{1,\ldots,n\}$ be the states of the optimized product,
   and let $(1,1)$ be initial.
\- For each state $(i,j)$:
   \li
   \-[--] for every $x \models \bar a_i \bar g_j$:\,
      $\delta((i,j), x) = \big(\{(i,j)\},2,\rej\big)$:\\
      there is a self-loop colored $0$ labelled $\bar a_i \bar g_j$;
   \-[--] for every $x \models a_i \bar g_j$:\,
      $\delta((i,j), x) = \big(\{(i',j)\mid i' \in \{1,\ldots,m\}\},1,\acc\big)$:\\
      on reading a letter satisfying $\bar g_j$,
      the system can choose a $1$-colored transition to any of
      $\big\{(i',j)\mid i' \in \{1,\ldots,m\}\big\}$; and
   \-[--] for every $x \models g_j$:\,
      $\delta((i,j), x) = \big(\{1,\ldots,m\}{\x}\{1,\ldots,n\},0,\rej \big)$:\\
      on reading a letter satisfying $g_j$,
      the environment can choose a $0$-colored transition to any state.
   \il
\il
The fixpoint formula for such COCOA product has the form:
\begin{align*}
\left[
  \begin{matrix}
    W_{1,1} \\ \smallvdots \\ W_{m,n}
  \end{matrix}
\right]
&=
\nu
\left[
  \begin{matrix}
    Z_{1,1} \\ \smallvdots \\ Z_{m,n}
  \end{matrix}
\right]\!\!.\,
\mu
\left[
  \begin{matrix}
    Y_{1,1} \\ \smallvdots \\ Y_{m,n}
  \end{matrix}
\right]\!\!.\,
\nu
\left[
  \begin{matrix}
    X_{1,1} \\ \smallvdots \\ X_{m,n}
  \end{matrix}
\right]\!\!.
\LTLsquare \LTLdiamond\!
  \left[
  \begin{matrix}
    \psi_{1,1} \\
    \smallvdots \\
    \psi_{m,n}
  \end{matrix}
  \right], \text{ where for all $i,j$:}\\
  \psi_{i,j} &=
  g_j (\!\!\!\!\!\!\!\!\bigwedge_{\scalebox{0.6}{\specialcellC{$i' \in \{1,\ldots,m\}$\\$j' \in \{1,\ldots, n\}$}}} \!\!\!\!\!\!\!\!\! Z_{i',j'})
  ~~\lor~~
  a_i \bar g_j (\!\!\!\!\!\!\!\!\bigvee_{i' \in \{1,\ldots,m\}} \!\!\!\!\!\!\!\!\!Y_{i',j})
  ~~\lor~~
  \bar a_i \bar g_j X_{i,j}
\end{align*}
The conjunction $\bigwedge_{i'\!,\,j'} Z_{i'\!,\,j'}$ and disjunctions $\bigvee_{i'} Y_{i'\!,\,j}$
enable faster information propagation which results in smaller number of fixpoint iterations.
Such information sharing is present in GR(1) fixpoint Equation~\ref{eqn:fixpoint-GR1},
and it is in this sense the COCOA approach generalizes GR(1) approach.
In contrast,
the parity fixpoint formula for GR(1) specifications misses this acceleration.

We now optimize the equation to reduce the number of variables.
First, we introduce variables $Y_j$ and $Z_j$, for $j \in \{1,...,n\}$, and
transform the formula into
\begin{equation*}
\left[
  \begin{matrix}
    W_1 \\ \smallvdots \\ W_n
  \end{matrix}
\right]
=
\nu
\left[
  \begin{matrix}
    Z_1 \\ \smallvdots \\ Z_n
  \end{matrix}
\right]\!\!.\,
\mu
\left[
  \begin{matrix}
    Y_1 \\ \smallvdots \\ Y_n
  \end{matrix}
\right]\!\!.\,
\left[
  \begin{matrix}
    \bigvee_{i} \Phi_{i,1} \\ \smallvdots \\ \bigvee_i \Phi_{i,n}
  \end{matrix}
\right]\!\!, \textit{ where}
\end{equation*}
\begin{equation*}
\left[
\begin{matrix}
  \Phi_{1,1} \\ \smallvdots \\ \Phi_{m,n}
\end{matrix}
\right]
=
\nu
\left[
  \begin{matrix}
    X_{1,1} \\ \smallvdots \\ X_{m,n}
  \end{matrix}
\right]\!\!.
\LTLsquare \LTLdiamond\!
  \left[
  \begin{matrix}
    \psi_{1,1} \\
    \smallvdots \\
    \psi_{m,n}
  \end{matrix}
  \right]\!\!, \textit{ where}
  \vspace{-2mm}
\end{equation*}
\begin{equation*}
  \psi_{i,j} =
  g_j (\!\!\!\!\!\!\!\!\bigwedge_{j' \in \{1,\ldots, n\}} \!\!\!\!\!\!\!\!\! Z_{j'})
  ~~\lor~~
  a_i \bar g_j Y_j
  ~~\lor~~
  \bar a_i \bar g_j X_{i,j}
\end{equation*}
Note that for every $i \in \{1,\ldots,m\}$,
the value $W_{i,j}$ computed by the old formula equals the value $W_j$ computed by the new formula ($W_{i,j} = W_i$),
where $j \in \{1,\ldots,n\}$.
We then introduce a fresh variable $Z$, and transform the formula to:
\begin{equation*}
W =
  \nu Z.\!\!\!\!\bigwedge_{j\in\{1,\ldots,n\}} \!\!\!\! \Psi_j,\textit{ where}
\end{equation*}
\begin{equation*}
\left[
  \begin{matrix}
    \Psi_1 \\ \smallvdots \\ \Psi_n
  \end{matrix}
\right]
=
\mu
\left[
  \begin{matrix}
    Y_1 \\ \smallvdots \\ Y_n
  \end{matrix}
\right]\!\!.\,
\left[
  \begin{matrix}
    \bigvee_{i} \Phi_{i,1} \\ \smallvdots \\ \bigvee_i \Phi_{i,n}
  \end{matrix}
\right]\!\!, \textit{ where}
\end{equation*}
\begin{equation*}
\left[
\begin{matrix}
  \Phi_{1,1} \\ \smallvdots \\ \Phi_{m,n}
\end{matrix}
\right]
=
\nu
\left[
  \begin{matrix}
    X_{1,1} \\ \smallvdots \\ X_{m,n}
  \end{matrix}
\right]\!\!.
\LTLsquare \LTLdiamond\!
  \left[
  \begin{matrix}
    g_1 Z ~\lor~ a_1 \bar g_1 Y_1 ~\lor~ \bar a_1 \bar g_1 X_{1,1} \\
    \smallvdots \\
    g_n Z ~\lor~ a_m \bar g_n Y_n ~\lor~ \bar a_m \bar g_n X_{m,n}
  \end{matrix}
  \right]
\end{equation*}
After this transformation,
we have $W = W_j$ for every $j \in \{1,\ldots,n\}$.
Finally, the last equations can be folded into the formula
\begin{equation*}
  W = \nu Z. \bigwedge_{j=1}^n \mu Y. \bigvee_{i=1}^m \nu X. \LTLsquare\LTLdiamond \big[g_j Z ~\lor~ a_i \bar g_j Y ~\lor~ \bar a_i \bar g_j X\big]
\end{equation*}
which is equal to Equation~\ref{eqn:fixpoint-GR1} modulo expressions in front of the variables.
Our prototype tool implements a generalized version of such formula optimization.

\section{Evaluation}
\label{sec:experiments}
%\ak{I still do not quite understand why there is a large efficiency diff between slugs and cocoa}

Evaluation goals are:
{\sf (G1)} show that standard LTL synthesizers do not fit our synthesis problem,
{\sf (G2)} compare our approach against specialized GR(1) synthesizer, and
{\sf (G3)} compare the COCOA approach against the parity approach.

We implemented COCOA and parity approaches in a prototype tool \reboot.
It uses SPOT \cite{DBLP:conf/atva/Duret-LutzLFMRX16} to convert LTL specifications (the liveness part of GR(1)) to deterministic parity automata.
From it, \reboot builds COCOA using the construction described in~\cite{DBLP:conf/fsttcs/EhlersS22}.
The COCOA is then compiled into a fixpoint formula in Equation~\ref{eqn:fixpoint-COCOA},
and symbolically evaluated on the game graph.
For symbolic encoding of game positions and transitions,
we use the BDD library CUDD~\cite{Somenzi98cudd:cu}.

We compare our approaches with GR(1) synthesis tool \slugs~\cite{DBLP:conf/cav/EhlersR16}
and the LTL synthesis tool \strix~\cite{meyer2018strix} which represent the state of the art.
The experiments were performed on a Linux machine with AMD EPYC 7502 processor;
the timeout was set to $1$ hour.
To implement the comparison, we collected existing and created new benchmarks:
AMBA, lift, and robot on a grid.
Each specification is written in an extension of the \slugs format: it encodes a symbolic game graph using logical formulas over system and environment propositions, and an LTL property on top of it.
In total, there are $80$ benchmarks, all realizable.

The evaluation data is available at \url{https://doi.org/10.5281/zenodo.10448487}

\paragraph{AMBA and lift.}
We use two parameterized benchmarks inspired by~\cite{DBLP:journals/jcss/BloemJPPS12},
each having two versions,
a GR(1) and an LTL version.
The first specification encodes an elevator behaviour and is parameterized by the number of floors.
Its GR(1) specification has one liveness assumption and a parameterized number of guarantees
($\LTLGF \to \bigwedge_i \LTLGF$).
Lift's LTL version adds an additonal request-response assumption and has the form
$\LTLGF \land (\LTLGF \to \LTLGF) \to \bigwedge_i \LTLGF$,
which requires $5$ parity colors.
There are $24$ GR(1) instances and $21$ LTL instances,
with the number of Boolean propositions ranging from $7$ to $34$.
The AMBA specification describes the behaviour of an industrial on-chip bus arbiter serving a parameterized number of clients.
Its GR(1) version has the shape $\LTLGF \to \bigwedge_i \LTLGF$;
our new LTL modification replaces one safety guarantee $\varphi$ by $\LTLF\LTLG\varphi$,
which allows the system to violate it during some initial phase,
and we add an assumption of the form $\LTLGF \to \LTLGF$.
Overall, the AMBA's LTL specification has the form
$\LTLGF \land (\LTLGF \to \LTLGF) \to \LTLFG \land \bigwedge_i \LTLGF$,
and requires $7$ parity colors.
There are $14$ GR(1) instances and $7$ LTL instances;
the number of Boolean propositions is $22$ for the specification serving two clients, and $77$ for the $15$-client version.
%{\ak{move into related work}\color{gray} Note that there is a symbolic synthesis approach to handle \emph{one}-pair Rabin acceptance~\cite{DBLP:conf/nfm/Ehlers11}(NOPE: it is more than that: it can handle gua as on the right side but not the ass as on the left side), but it does not apply to two pairs and it was not implemented in a tool so far.}

\paragraph{Robot on a grid.}
This benchmark describes the standard scenario from robotics domain:
a robot moves on a grid,
there are walls, doors, pickup and delivery locations, and a moving obstacle.
When requested, the robot has to pickup a package and deliver it to the target location,
while avoiding collisions with the walls and the obstacle and passing through the doors only when they are open.
The GR(1) specification has parameterized number of assumptions and guarantees:
$\bigwedge_i \LTLGF \to \bigwedge_i \LTLGF$.
The LTL version introduces preferential paths:
the robot has to eventually always use it assuming that the moving obstacle only moves along her preferred path.
This yields the shape
$\LTLFG\land\bigwedge_i \LTLGF \to \LTLFG\land\bigwedge_i\LTLGF$ ($5$ colors).
There are $16$ maps of size $8{\x}16$ with varying number of delivery-pickup locations and doors.
The number of Boolean propositions ranges from $24$ to $53$.

% Ayrat:
% LTLf benchmarks are too easy and make no sense.
%\paragraph{LTLf specifications.}
%Finally, we use the benchmarks from LTL synthesis on finite traces \cite{XXX}.\ak{todo}
%Our tool \reboot outperforms 2SLS from \cite{} for the same reason it outperforms \strix:
%we manually translate the safety part of the specification into a symbolic arena,
%whereas 2SLS uses logic-to-automata translation tool MONA for that purpose,
%and it represents the bottleneck.
%This comparison signifies yet another time the importance of either
%starting with symbolic arenas or having an efficient construction procedure.
%\ak{todo}

\paragraph{G1: Comparing with LTL synthesizer.}
Figure~\ref{fig:eval} shows a cactus plot.
On these problems,
the LTL synthesizer \strix is slower than specialized solvers.
The reason is the sheer number of states in benchmark game arenas:
e.g., benchmark $\mathit{amba15}$ uses 77 Boolean propositions,
yielding the naive estimate of game arena size in $2^{77}$ states.
Solver \strix tries to construct an explicit-state automaton describing this game arena and the LTL property, which is a bottleneck.
In contrast, symbolic solvers like \slugs or \reboot represent game arenas symbolically using BDDs, and \reboot constructs explicit automata only for LTL properties.

\begin{figure}[tb]
  \center
  \includegraphics[width=0.23\columnwidth]{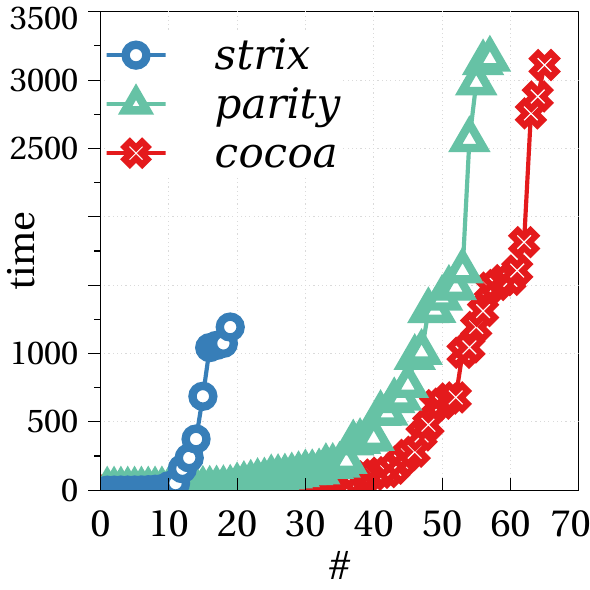}~~
  \includegraphics[width=0.23\columnwidth]{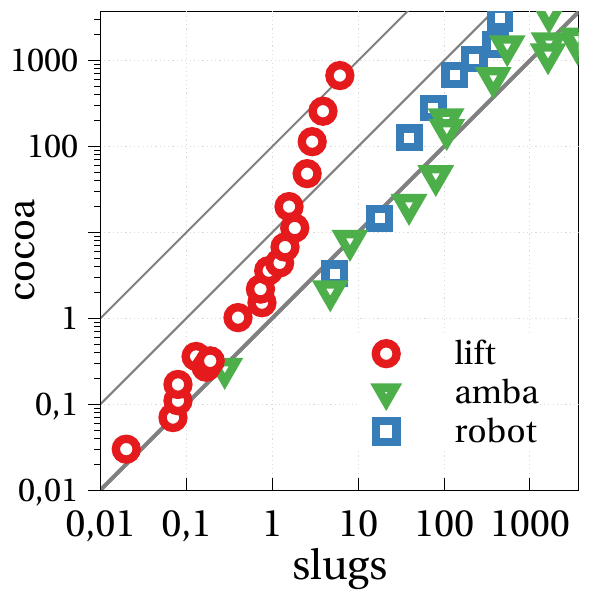}~~
  \includegraphics[width=0.23\columnwidth]{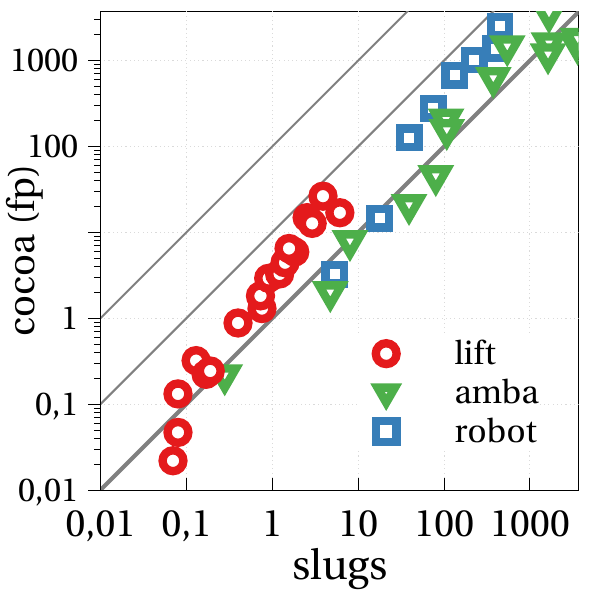}~~
  \includegraphics[width=0.23\columnwidth]{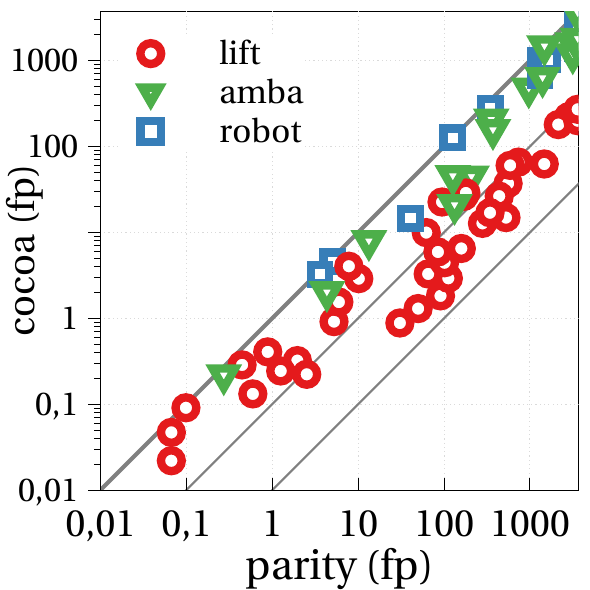}~
  \vspace{-2mm}
  \caption{From left to right: (G1) Cactus plot comparing our approaches with LTL synthesizer \strix~\cite{meyer2018strix};
  (G2a) Comparing COCOA-based approach with GR(1) synthesizer \slugs~\cite{DBLP:conf/nfm/Ehlers11};
  (G2b) The same but excluding LTL-to-parity translation time;
  (G3) Comparing COCOA and parity approaches (excluding LTL-to-parity translation time).}
\label{fig:eval}
\end{figure}

\paragraph{G2: Comparing with GR(1) synthesizer.}
The second diagram in Figure~\ref{fig:eval} compares the COCOA approach with \slugs on the GR(1) benchmarks.
The diagram shows the total solving time,
including the time \reboot spends calling SPOT for translating GR(1) liveness formula to parity automaton.
On Lift examples,
most of the time is spent in this translation when the number of floors exceeds $15$:
for instance, on benchmark $\mathit{lift20}$ \reboot spent $650$ out of total $670$ seconds in translation.
If we count only the time spent in fixpoint evaluation --
and that is a more appropriate measure since GR(1) liveness formulas have a fixed structure --
the performances are comparable, see the third diagram.

\paragraph{G3: COCOA vs.\ parity.}
The last diagram in Figure~\ref{fig:eval} compares COCOA and parity approaches on all the benchmarks,
and shows that the COCOA approach is significantly faster than the parity one.
We note that on these examples,
the number of states in the optimized COCOA product was equal to or less than the number of states in the parity automaton.
At the same time,
the number of fixpoint iterations performed by the COCOA approach was always significantly smaller than for the parity one.
Intuitively,
this is due to the structure of COCOA fixpoint equation that propagates information faster than the parity one.

\paragraph{Remarks.}
We did not compare with other symbolic approaches for solving parity or Rabin games~\cite{oink,di2018solving,banerjee2023fast}:
although they use symbolic algorithms,
as input these tools require games in explicit form or
their game encoding separates positions into those of player-1 and player-2;
both significantly affects the performance.

While all our benchmarks were realizable, the prototype tool was systematically compared against other approaches on both realizable and unrealizable random specifications using fuzz testing.

\section{Conclusion}
In this paper, we presented an approach for realizability checking on game graphs for the full set of omega-regular specifications. The approach utilizes chains of co-Büchi automata as intermediate representation of the specification as they can be minimized and made canonical in polynomial time.
It is the first practical application of chains of co-Büchi automata.
We showed that when building a fixed point equation from specifications given as such chains,
we naturally obtain the Generalized Reactivity(1) synthesis fixpoint formula as a special case.
This also holds for specifications that are only semantically equivalent to GR(1) specifications, due to the canonicity of the chain representation.
As a consequence, our approach bridges GR(1) synthesis and full LTL/omega-regular synthesis, and our experiments show that it is attractive for specifications that go slightly beyond GR(1), and definitively outperforms standard approaches.

The computed fixpoint equations have the nice property that they are also applicable beyond Boolean or finite game graphs.
For instance, they could be used for computing control strategies for discrete-time linear hybrid systems using a AND-Inverter graph representation with linear inequalities for representing state sets \cite{DBLP:conf/atva/DammDHJPPSWW07}.

We focused on realizability checking in this paper and deferred the extraction of implementations from the computed fixpoints to future work.
Another future direction is further optimizing the chain product construction,
by removing the transitions forbidden by the game graph.
Finally,
the fact that on some benchmarks the LTL-to-parity translation of the liveness specification takes a lot of time
indicates that it might be helpful to build COCOA compositionally,
or alternatively, to directly build COCOA from LTL.
% A thorough examination of the possibilities is beyond the scope of this paper, though.

\bibliographystyle{splncs04}
\bibliography{bib.bib}
\end{document}